\documentclass[11pt]{article} 

\usepackage[utf8]{inputenc} 

\usepackage{authblk}
\usepackage{amsmath}
\usepackage{amsthm,amsfonts,amssymb,txfonts,pxfonts,empheq}
\usepackage{hyperref}
\usepackage{hyperref}
\usepackage{fullpage}
\usepackage{framed}
\usepackage{comment}

\usepackage{mathtools,xparse}

\newtheorem{proposition}{Proposition}[section]
\newtheorem{lemma}[proposition]{Lemma}
\newtheorem{corollary}[proposition]{Corollary}

\newtheorem{theorem}[proposition]{Theorem}
\newtheorem{remark}[proposition]{Remark}

\newtheorem {RHP}{Riemann-Hilbert problem}

\newtheorem*{theorem*}{Theorem}
\newtheorem*{assumption*}{Working assumption}
\newtheorem*{MDBP}{Meromorphic $\overline{\partial}$- problem}
\newtheorem*{MRHP}{Model Riemann-Hilbert problem}
\newtheorem*{SDBP}{Smooth $\overline{\partial}$- problem}
\def\res{\mathop{Res}}

\DeclarePairedDelimiter{\norm}{\lVert}{\rVert}
\DeclarePairedDelimiter{\abs}{\lvert}{\rvert}

\newcommand{\lp}{\left(}
\newcommand{\rp}{\right)}
\newcommand{\ls}{\left[}
\newcommand{\rs}{\right]}
\newcommand{\lb}{\left\{}
\newcommand{\rb}{\right\}}
\newcommand{\la}{\left|}
\newcommand{\ra}{\right|}

\newcommand{\de}{\textrm{d}}
\newcommand{\dbar}{\overline{\partial}}

\newcommand{\lird}{\textrm{L}^\infty ( \mathbb{R}^2 )}

\newcommand{\bfm}{\mathbf{m}}
\newcommand{\bfmt}{\tilde{\mathbf{m}}}
\newcommand{\bfe}{\mathbf{e}}
\newcommand{\bbR}{\mathbb{R}}
\newcommand{\bbC}{\mathbb{C}}

\newcommand{\bbJ}{\mathbb{J}}

\newcommand{\mt}{\tilde{m}}

\newcommand{\sigunonum}{\left( \begin{array}{cc} 0 & 1\\ 1 & 0 \end{array} \right)}

\newcommand{\lnorm}

\newcommand{\bfA}{\mathbf{A}}

\newcommand{\ddx}{\frac{\partial}{\partial x}}

\newcommand{\ddu}{\frac{\partial}{\partial u}}

\newcommand{\dda}{\frac{\partial}{\partial a}}

\newcommand{\imag}{\text{Im}}

\newcommand{\ldru}{L^2\lp \bbR, \de a \rp}

\newcommand{\bt}{\tilde{b}}

\newcommand{\sxt}{\lp s;x,t \rp}

\newcommand{\ta}{\tilde{a}}

\newcommand{\xovt}{\frac{x}{t}}

\date{}

\title{Long-time asymptotic analysis of the Korteweg-de Vries equation via the dbar steepest descent method: \\The Soliton region}

\newsavebox\affbox

\author[1]{Pietro Giavedoni\thanks{\href{mailto:Pietro.Giavedoni@univie.ac.at}{Pietro.Giavedoni@univie.ac.at}; \href{mailto:addenaro@gmail.com}{addenaro@gmail.com}.}}

\begin{document}

\maketitle

\begin{abstract}
We address the problem of long-time asymptotics for the solutions of the Korteweg-de Vries equation under low regularity assumptions. We consider decreasing initial data admitting only a finite number of moments. For the so-called ``soliton region'', an improved asymptotic estimate is provided, in comparison with the one in \cite{GT}. Our analysis is based on the dbar steepest descent method proposed by P. Miller and K. T. D. -R. McLaughlin.  
\end{abstract}

\section{Introduction}

Let us consider the initial-boundary value problem for the Korteweg-de Vries (KdV) equation with a decaying initial datum
\begin{align}\label{CP}
\lb\begin{array}{c}
q_{t}(x,t) = 6q(x,t)q_x(x,t) - q_{xxx}(x,t)\\
q\lp x,t=0 \rp = q_0\lp x \rp, \quad q_0\lp x \rp \rightarrow 0 \,\,\, \mbox{as} \,\,\, \abs*{x} \rightarrow \infty.
\end{array}\right.
\end{align}
Existence and uniqueness of real-valued, classical solutions can be proved via the inverse scattering transform, introduced by Green, Gardner, Kruskal and Miura in their seminal work \cite{GGKM}. The long time behavior of these last ones has been extensively investigated in the literature (\cite{ZabKru}\cite{AblNew}\cite{Manakov}\cite{ZakMan}\cite{ShaOnt}\cite{TanKor}\cite{SegAbl}\cite{Deift}). The solutions are known to eventually decompose into a certain number of solitons, travelling to the right, plus a radiation part, propagating to the left. In this paper we wish to consider the so-called \emph{soliton region}, formed by those points of the $(x,t)$-plane satisfying $\tfrac{x}{t}\geq C_0$, for some fixed constant $C_0>0$. In order to detail more about existing results, let us recall that the solutions of (\ref{CP}) are uniquely individuated by the scattering data of the operator 
\begin{align} 
H:= -\frac{\de}{\de x^2} + q_0\lp x \rp.
\end{align}
associated with the initial datum. These last ones consist of a finite number of eigenvalues, $-\kappa_1^2,-\kappa_2^2,\ldots,-\kappa_M^2$, with $0<\kappa_1<\kappa_2<\ldots<\kappa_M$, of the corresponding norming constants $\gamma_j>0$, and of the reflection coefficient $r:\bbR\rightarrow\bbC$. The long time asymptotics of solutions of (\ref{CP}) in the soliton region reads as follows
\begin{align}\label{Nuova_sol}
q\lp x,t \rp = -2\sum_{j=1}^M \frac{\kappa_j^2}{\cosh^2\lp \kappa_j x - 4\kappa_j^3 t - p_j \rp} +\mathcal{E}\lp x,t \rp, \quad\quad t\rightarrow +\infty,\,\,\,\tfrac{x}{t}\geq C_0 .
\end{align} 
Here the phase-shifts are given by
\begin{align}
p_j = \frac{1}{2} \log\ls \frac{\gamma_j^2}{2\kappa_j}\prod_{l=j+1}^M \lp \frac{\kappa_l-\kappa_j}{\kappa_l+\kappa_j} \rp^2 \rs.
\end{align}
The term $\mathcal{E}(x,t)$ in (\ref{Nuova_sol}) is know to be small for large $t$, its magnitude depending on the smoothness and decay properties of $q_0(x)$. This formula was established by Hirota \cite{HirExa}, Tanaka \cite{TanOnt} and Wadati and Toda \cite{WadTod} independently, for vanishing reflection coefficient $r$. The general case was first treated by Tanaka \cite{TanKor} and Shabat (\cite{ShaOnt}).
More recently, Grunert and Teschl proved such asymptotic behavior for initial data with lower regularity \cite{GT}. Their approach relies on the steepest descent analysis of Riemann-Hilbert problem \ref{Cesare} (see next section) via a decomposition of the nonanalytic reflection coefficient $r$ into an analytic approximant and a small rest. In this paper we examine the same Riemann-Hilbert problem via the modern dbar method, introduced by Miller and McLaughlin in \cite{Circle} and \cite{Line}. In particular, we establish a better estimate of $\mathcal{E}\lp x,t \rp$ for a larger class of initial data. Our main result is the following 
\begin{theorem}\label{Hauptsatz}
Let the reflection coefficient $r$ associated to the initial datum $q_0$ belong to $\mathcal{C}^{N+1}\lp \bbR \rp$, for some integer $N\geq 1$. Assume that $r$ and its first $N$ derivatives tend zero at $\pm \infty$. Moreover, let $r^{(N+1)}$ belong to the Wiener algebra on the real line. That is, assume that this last one is the image, via Fourier transform, of some function in $L^1\lp \bbR \rp$. Fix $C_0>0$. Then there exists a constant $C$ such that
\begin{align}
\abs*{\mathcal{E}\lp x,t \rp} \leq C t^{-N- \frac{4}{3}}
\end{align}
for all $t$ sufficiently large and $x\geq C_0 t$.
\end{theorem}
Notice that these hypotheses are satisfied by all initial data admitting $N+2$ moments (\cite{GT}):
\begin{align}
\int_{-\infty}^{+\infty}\lp 1+\abs*{x}^{N+2} \rp \abs*{q_0\lp x \rp} \de x < \infty.
\end{align}  
Theorem \ref{Hauptsatz} is achieved via a careful treatment of the imaginary part of the phase $\Phi$ defined in (\ref{def_Phi}). After the standard procedure, the jump (\ref{salto}) of the Riemann-Hilbert problem \ref{Cesare} across the real axis is decomposed and displaced partly below and partly above it. Accordingly, $\Phi$ develops a non-vanishing real part, providing a decay of the decomposed jumps towards the identity matrix. The novelty here is that subsequently we reconsider the oscillations originating from the imaginary part of $\Phi$. From their analysis, we extract additional information about the decay of the error term $\mathcal{E}$, corresponding exactly to our improvement of the estimates. To our best knowledge, the idea of this last step is new in the literature. \\
Further advantages of our approach are the following. First of all, our analysis requires less sophisticated technical means, employing basically calculus at an undergraduate level and the Van der Corput lemma. Using them we provide simpler and more explicit expressions for $\mathcal{E}$ (see formulas (\ref{Ciliegio}), (\ref{Saldini}) and the following ones in section \ref{Nuova_ossessione}). These ones can be easily employed for a more detailed, long time asymptotic expansion including higher order corrections. Moreover, they might turn out to be useful for the analysis of analogue Riemann-Hilbert problems beyond the framework of integrability. This a current research interest of ours.


\section{Proof of the result}

Let us fix an integer $N\geq 1$ and a constant $C_0>0$. We assume $x\geq C_0 t$ for the remaining part of the paper and prove theorem \ref{Hauptsatz}. The hypotheses on the reflection coefficient are also understood to hold, without further recalling them. Analogously to \cite{GT}, we produce the solution of KdV corresponding to the initial datum $q_0(x)$ - or equivalently, to the associated scattering data $\kappa_j$, $\gamma_j$ and $r(w)$ - via the following

\begin{RHP}\label{Cesare}
Find a function $\bfm\lp w \rp = ( m_1(w),  m_2(w) )$ meromorphic away from the real axis, with simple poles at $\pm \imath \kappa_1, \pm \imath\kappa_2,\ldots,\pm\imath\kappa_M$, satisfying:
\begin{description}
\item{i\_} "Jump condition ".
For every $w\in \bbR$ one has 
\begin{align}\label{cond_salti}
\bfm_{+}\lp w \rp = \bfm_{-}\lp w \rp V\lp w \rp, 
\end{align}
where
\begin{align}\label{salto}
V\lp w \rp = \lp\begin{array}{cc}
1 - \la r\lp w \rp \ra^2 & -\overline{r\lp w \rp}e^{-t\Phi\lp w \rp}\\
r\lp w \rp e^{t\Phi\lp w \rp} & 1
\end{array}\rp.
\end{align}

\item[ii\_]  "Residue condition" 
\begin{align} 
\res_{w=\imath\kappa_j} \bfm\lp w \rp &= \lim_{w\rightarrow \imath \kappa_j}\bfm\lp w \rp\lp \begin{array}{cc} 
0 & 0 \\
\imath\gamma_j^2e^{t\Phi\lp \imath \kappa_j\rp} & 0
\end{array} \rp\\
 \res_{w = -\imath\kappa_j}\bfm\lp w \rp &= \lim_{w\rightarrow -\imath \kappa_j}\bfm\lp w \rp\lp\begin{array}{cc}
0 & -\imath\gamma_j^2 e^{t\Phi\lp \imath \kappa_j \rp}\\
0 & 0
\end{array}\rp
\end{align}
for $j=1,2,\ldots, M$.

\item[iii\_] "Symmetry condition"
\begin{align} \label{WaitingFor}
\bfm\lp -w \rp = \bfm\lp w \rp\lp  \begin{array}{cc} 0 & 1 \\ 1 & 0 \end{array} \rp
\end{align}

\item[iv\_] "Normalization condition"
\begin{align} \label{Fellini}
\lim_{\abs{w} \rightarrow + \infty}\bfm\lp w \rp = \lp \begin{array}{cc} 1 & 1 \end{array} \rp
\end{align}
\end{description}
Here the phase is given by 
\begin{align} \label{def_Phi}
\Phi\lp w \rp = 8\imath w^3 + 2\imath w \frac{x}{t}.
\end{align}
\end{RHP}

The solution of the system (\ref{CP}) at an arbitrary time $t>0$ is then recovered by the formula
\begin{align} \label{Sava}
q\lp x,t \rp = -2i\frac{\de}{\de x}m_1^{\lp 1 \rp}\lp x,t \rp.
\end{align}
the right-hand side being computed from the expansion
\begin{align} \label{Tisa}
\bfm\lp w;x,t \rp =\bfm^{\lp 0 \rp}\lp x,t \rp + \frac{\bfm^{\lp 1 \rp}\lp x,t \rp}{w} + \mathcal{O}\lp \frac{1}{w^2} \rp, \quad\quad w\rightarrow \infty.
\end{align}

Our proof of theorem \ref{Hauptsatz} consists in subsequent reformulations of the Riemann-Hilbert problem, till obtaining a convenient, equivalent integral equation defined on the plane.

\subsection{Non-holomorphic extensions of the relection coefficient. Reformulation of the jump across the real axis.}

In this section we remove the jump of $\bfm$ across the real axis, exchanging it for some non-analytic behaviour on a strip around this last one. Let us fix the parameter 
\begin{align}
\delta := \min\lb \frac{\kappa_1}{100}, \frac{\sqrt{C_0}}{100} \rb.
\end{align}
We define the following regions 
\begin{align}
\Omega_1 &= \lb w\in \bbC \mbox{ such that } 0\leq \imag\, w < \delta \rb\\
\Omega_2 &= \lb w\in \bbC \mbox{ such that }  \delta\leq \imag\, w < 2\delta\rb\\
\Omega_3 &= \lb w\in \bbC \mbox{ such that }  2\delta \leq \imag\, w \rb
\end{align}
\begin{figure}
\begin{center}
\includegraphics[width= 0.714\textwidth]{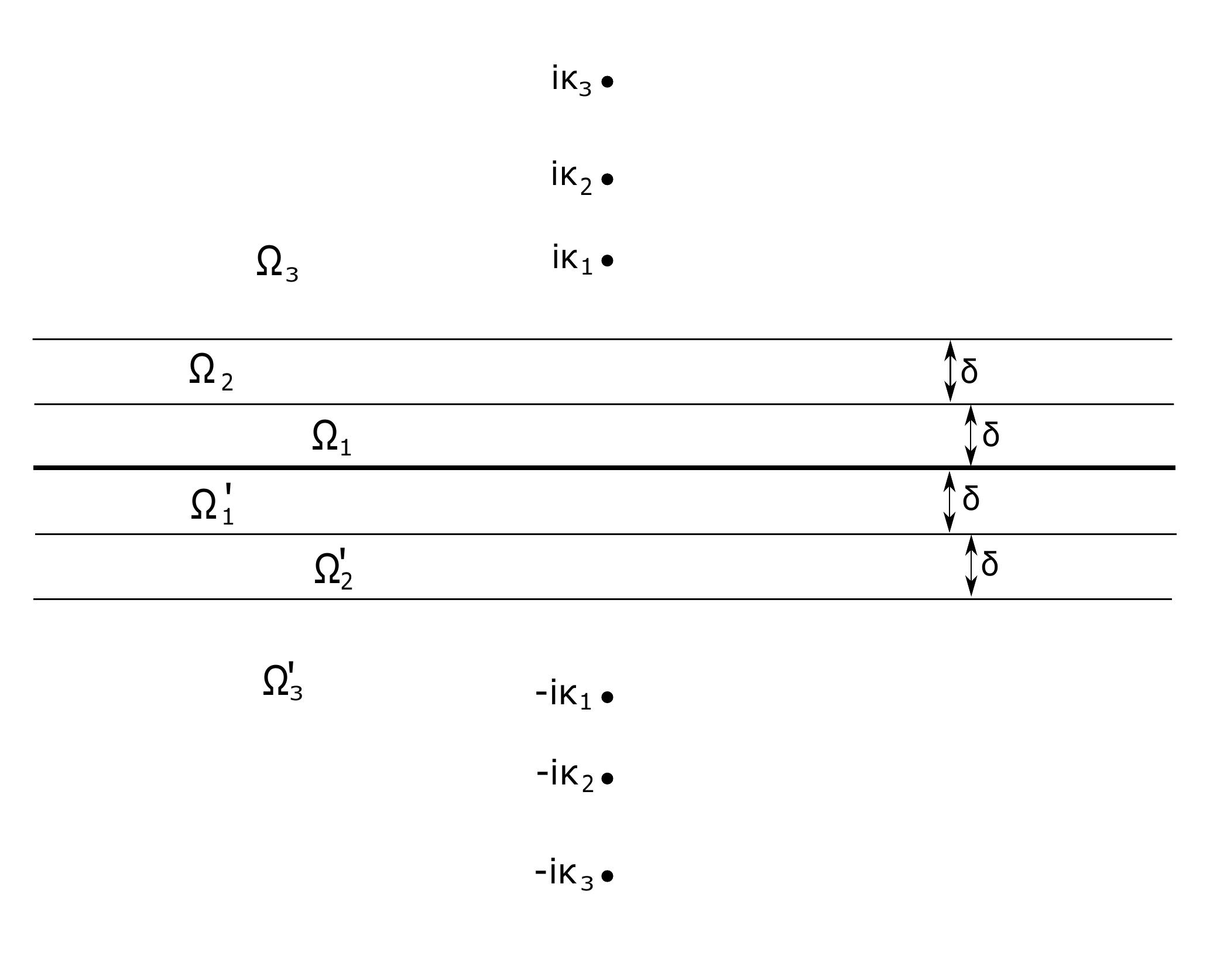}
\caption{\label{figura_striscie}{\it The strips introduced close to the real axis.}}
\end{center}
\end{figure} 
With $\Omega_1^{\prime}, \Omega_2^{\prime}$ and $\Omega_3^{\prime}$ we will indicate the corresponding reflected strips w.r.t. the real axis (see figure \ref{figura_striscie}). Let us also fix the notation 
\begin{align}
w=u+iv, \quad\quad s=a+ib, \quad\quad\quad a,b,u,v \in \bbR.
\end{align}
for the remaining part of the paper. We wish to consider the following non-analytic extension of the reflection coefficient $r$ to the whole upper-half plane:
\begin{align}\label{def_R}
R\lp w \rp = R\lp u+iv \rp := \ls r\lp u \rp + r^{\prime}\lp u \rp\lp iv \rp + \frac{1}{2}r^{\prime\prime}\lp u \rp\lp iv \rp^2 + \ldots + \frac{1}{N!}r^{(N )}\lp u \rp\lp iv \rp^N \rs \cdot \chi\lp \frac{v}{\delta} \rp.
\end{align}
Here $\chi$ is chosen as follows\footnote{The choice of $\chi$ is by no means unique. Any other smooth function with analogue ``cut-off'' properties would do.} 
\begin{align}\label{def_chi}
\chi\lp v \rp = \lb \begin{array}{ll}
1 & 0\leq v < 1\\
\exp{\ls\frac{\lp v-1 \rp^2}{\lp v-1 \rp^2 - 1}\rs} & 1\leq v <2\\
0 & v \geq 2
\end{array}
\right.
\end{align}
Notice that $R(w)$ vanishes on $\Omega_3$. Moreover, there exists a constant $C>0$ such that
\begin{align}\label{diseg_R}
\abs*{\dbar R\lp w \rp} \leq C v^N, \quad\quad \quad w\in \Omega_1\cup\Omega_2.
\end{align}
where 
\begin{align}
\dbar := \frac{1}{2}\lp \ddu + i \frac{\partial}{\partial v} \rp
\end{align}
These are the main properties motivating our choice of such extension. Using (\ref{def_R}), one can decompose the jump matrix $V$ as follows
\begin{align}
V\lp w \rp = A_{low}\lp w \rp A_{upp}\lp w \rp^{-1}, \quad\quad w\in \mathbb{R}.
\end{align}
Here
\begin{align}
A_{upp}\lp w \rp = \lp \begin{array}{cc}
1 & 0 \\
-R\lp w \rp e^{t\Phi\lp w \rp} & 1
\end{array}\rp \quad \quad\quad\quad \imag\lp w \rp \geq 0
\end{align}
and 
\begin{align} 
A_{low}\lp w \rp= \lp\begin{array}{cc}
1 & -R\lp -w \rp e^{-t\Phi\lp w \rp}\\
0 & 1
\end{array}\rp \quad\quad\quad\quad \imag\lp w \rp \leq 0.
\end{align}
Mimicking the classical nonlinear steepest descent method (\cite{Deift},\cite{GT}), we introduce
\begin{align}\label{def_m_tilde}
\bfmt\lp w \rp = \lb\begin{array}{cc}
\bfm\lp w \rp A_{upp}\lp w \rp & \imag\lp w \rp\geq 0\\
\bfm\lp w \rp A_{low}\lp w \rp &  \imag\lp w \rp \leq 0
\end{array}\right.
\end{align}
Riemann-Hilbert problem \ref{Cesare} for $\bfm$ is then equivalent to the following
\begin{MDBP} \label{dbar_problem_per_m_tilde}
Find a two dimensional, vector-valued function $\bfmt = (\mt_1,\mt_2)$ continuous on $\bbC\backslash\lb \pm i\kappa_1,\pm i\kappa_2, \ldots, \pm i \kappa_n \rb$ and differentiable with continuity as a function of two real variables away from the real axis, such that
\begin{description}
\item[\textbf{i.}] ``$\dbar$-condition''
\begin{align}
\dbar \bfmt\lp w \rp = \lb \begin{array}{lr}
\bfmt\lp w \rp\lp \begin{array}{cc} 
0 & 0 \\
-\dbar R\lp w \rp e^{-t\Phi\lp w \rp} & 0
\end{array}\rp & \imag\lp w \rp \geq 0\\
\bfmt\lp w \rp \lp \begin{array}{lr}
0 & \dbar R\lp -w \rp e^{-t\Phi\lp w \rp}\\
0 & 0
\end{array}\rp & \imag\lp w \rp \leq 0
\end{array}
\right. .
\end{align}
In particular, $\bfmt$ is holomorphic on $(\Omega_3 \cup \Omega_3^{\prime}) \backslash \lb \pm i\kappa_1, \pm i\kappa_2,\ldots, \pm i \kappa_n \rb.$
\item[\textbf{ii.}] ``Residue condition''. The vector-valued function $\bfmt$ has simple poles at $\lb \pm i\kappa_1, \pm i \kappa_2,\ldots ,\pm i \kappa_n\rb$, where it satisfies  
\begin{align} 
&\res_{w=\imath\kappa_j} \bfmt\lp w \rp = \lim_{w\rightarrow \imath \kappa_j}\bfmt\lp w \rp\lp \begin{array}{cc} 
0 & 0 \\
\imath\gamma_j^2e^{t\Phi\lp \imath \kappa_j\rp} & 0
\end{array} \rp\\
& \res_{w = -\imath\kappa_j}\bfmt\lp w \rp = \lim_{w\rightarrow -\imath \kappa_j}\bfmt\lp w \rp\lp\begin{array}{cc}
0 & -\imath\gamma_j^2 e^{t\Phi\lp \imath \kappa_j \rp}\\
0 & 0
\end{array}\rp 
\end{align}
for $j=1,2,\ldots, n$.
\item[\textbf{iii.}] ``Symmetry condition''
\begin{align}
\bfmt\lp -w \rp = \bfmt\lp w \rp \lp \begin{array}{cc} 0 & 1 \\ 1 & 0 \end{array} \rp, \quad\quad\quad w\in\bbC.
\end{align}
\item[\textbf{iv.}] ``Normalization condition''
\begin{align}
\lim_{\abs{w}\rightarrow + \infty} \bfmt\lp w \rp = (1,1).
\end{align}
\end{description}
\end{MDBP}

\begin{remark} \label{Remark_Claudia}
The solution $\bfmt$ of the meromorphic $\dbar$-problem needs actually to be differentiable also on the real axis, although possibly not with continuity. This is easily deduced, using (\ref{diseg_R}), from the representation
\begin{align}
\bfmt\lp w \rp = \frac{1}{2\pi i}\int_{\partial R} \frac{\bfmt\lp s \rp}{s-w} \de s - \frac{1}{\pi}\iint_{R}\frac{\dbar \bfmt \lp s \rp}{s-w} \de A\lp s \rp.
\end{align}
This is nothing else than the generalization of the Cauchy integral formula for smooth functions \cite{HorLin}. Here $R$ is understood to be a small, compact rectangle containing the point $w$, which for our purposes is chosen on the real line.
\end{remark}

\subsection{The model Riemann-Hilbert problem}

Our next goal is to remove the poles from vector $\bfmt$. To this purpose, we introduce in this section a model Riemann-Hilbert problem. An explicit expression of its solution won't be necessary for our analysis, but only some of its elementary properties concerning regularity and asymptotic behavior. These last ones are provided by proposition \ref{Roseto}.  

\begin{MRHP} 
Find a two times two matrix valued meromorphic function $M$  whose only poles are simple and lie in $\pm \imath \kappa_1,\pm\imath \kappa_2,\ldots,\pm \imath \kappa_M$, satisfying:
\begin{description}
\item[ii\_]  "Residue condition" 
\begin{align} \label{Resuno}
&\res_{w=\imath\kappa_j} M\lp w \rp = \lim_{w\rightarrow \imath \kappa_j}M\lp w \rp\lp \begin{array}{cc} 
0 & 0 \\
\imath\gamma_j^2e^{t\Phi\lp \imath \kappa_j\rp} & 0
\end{array} \rp\\
& \res_{w = -\imath\kappa_j}M\lp w \rp = \lim_{w\rightarrow -\imath \kappa_j}M\lp w \rp\lp\begin{array}{cc}
0 & -\imath\gamma_j^2 e^{t\Phi\lp \imath \kappa_j \rp}\\
0 & 0
\end{array}\rp \label{Resdue}
\end{align}
for $j=1,2,\ldots, M$.
\item[iii\_] "Symmetry condition"
\begin{align} \label{simm_M}
M\lp -w \rp = M\lp w \rp \sigunonum
\end{align}
\item[iv\_] "Normalization condition"
\begin{align}\label{norm_matricial}
M\lp w \rp = \lp\begin{array}{cc}
1 & 1 \\ - \imath w & \imath w
\end{array}\rp \lp \mbox{Id} + \frac{\imath H}{w}\lp \begin{array}{cc} 1 & 0 \\ 0 & -1 \end{array} \rp  + \mathcal{O}\lp \frac{1}{w^2} \rp\rp
\end{align}
for some constant $H$ possibly depending on $x$ and $t$.
\end{description}
\end{MRHP}

Let us remark that one cannot normalize $M$ imposing it to approach the identity matrix at infinity. No solution would then exist for a set of exceptional points $(x,t)$ which accumulate in the neighborhood of the peaks of the solitons as $t\rightarrow \infty$ (see \cite{BeaDei}, Chap. 38 for more details about this kind of issues). Normalization (\ref{norm_matricial}) will do for our purposes. All the information we need about the Model Riemann-Hilbert problem is contained in the following
\begin{proposition}\label{Roseto}
There exists a unique solution to the Model Riemann-Hilbert problem. This last one has the form
\begin{align}
M\lp w \rp = \ls\begin{array}{cc}
f\lp w \rp & f\lp -w \rp \\
g\lp w \rp & g\lp -w \rp
\end{array}\rs
\end{align}
where 
\begin{align} \label{Aneins}
f\lp w \rp = 1 + \frac{iA_1\lp x,t \rp}{w-i\kappa_1} + \frac{iA_2\lp x,t \rp}{w - i\kappa_2} + \ldots + \frac{iA_M\lp x,t \rp}{w-i\kappa_M}
\end{align}
and 
\begin{align}\label{Anzwei}
g\lp w \rp = -iw + H\lp x,t \rp + \frac{iB_1\lp x,t \rp}{w-i\kappa_1} + \frac{iB_2\lp x,t \rp}{w-i\kappa_2} + \ldots +\frac{iB_M\lp x,t \rp}{w-i\kappa_M}.
\end{align}
The functions $A_j\lp x,t \rp$ and $B\lp x,t \rp$ and their derivative w.r.t. $x$ are bounded in the whole $(x,t)$-plane. The constant $H$ is determined by conditions $\mathbf{ii}$ and $\mathbf{iii}$, and has the following asymptotic behaviour
\begin{align}\label{clas}
\ddx H\lp x,t \rp = - \sum_{j=1}^{M} \frac{k_j^2}{\cosh^2\lp \kappa_j x - 4\kappa_j^3 t - p_j \rp} + \mathcal{O}\lp e^{-Ct} \rp,\quad\quad t\rightarrow +\infty.
\end{align}
Here $C$ is a positive constant and 
\begin{align}\label{sic}
p_j := \frac{1}{2}\log\ls \frac{\gamma_j^2}{2\kappa_j}\prod_{l=j+1}^{M}\lp \frac{\kappa_l - \kappa_j}{\kappa_l + \kappa_j} \rp^2 \rs,\quad\quad\quad j=1,2,\ldots ,M.
\end{align}
\end{proposition}
\begin{proof}
Ansatz (\ref{Aneins}) and (\ref{Anzwei}) follow from points $\mathbf{iii}$ and $\mathbf{iv}$ of the Model Riemann-Hilbert problem. The ``Residue conditions'' translate, concerning the vector $\bfA$, into a system of linear equations
\begin{align} \label{sist_orig}
M\cdot \bfA = \boldsymbol{1}; \quad\quad\quad\quad M = Q+D. 
\end{align}
Here $\boldsymbol{1}$ indicates the M-dimensional column vector whose entries are all one. The matrices $Q$ and $D$, respectively symmetric and diagonal, are given by
\begin{align}
Q_{jl} := \frac{1}{\kappa_j + \kappa_l}, \quad\quad D_{jl} :=  \frac{\delta_{jl}}{C_l};\quad\quad\quad\quad j,l=1,2,\ldots M.
\end{align}
The constants $C_j$'s are defined as follows:
\begin{align}
C_j := \gamma_j^2 e^{t\Phi\lp i\kappa_j \rp},\quad\quad l,j=1,2\ldots M.
\end{align}
They vary between zero and $+\infty$ for $x$ and $t$ real. The matrix $Q$ is easily proved to be positive definite. Consequently also $M$ is, for all $x$ and $t$ real, and the system (\ref{sist_orig}) has a unique solution. 
Now, by elementary calculations the inverse of $M$ is shown to be entry-wise bounded as the entries of the diagonal matrix $D$ vary between zero and $+\infty$. This proves that also $\bfA$ is bounded.
Differentiating (\ref{sist_orig}) w.r.t. $x$, one obtains 
\begin{align} \label{sist_der}
M\bfA_x = -D_x \bfA
\end{align}
where
\begin{align} \label{collegamento}
\lp D_x \rp_{lj} = \frac{2\kappa_j}{C_j}\delta_{lj} = 2\kappa_j D_{lj},\quad\quad l,j=1,2,\ldots,M.
\end{align}
A solution $\bfA_x$ for this system exists and is unique. Rewriting (\ref{sist_orig}) as 
\begin{align}
D\bfA = \boldsymbol{1} - Q\bfA
\end{align}
the right-hand side is evidently bounded. So, in view of (\ref{collegamento}), also $D_x\bfA$ is. This results then into boundedness for $\bfA_x$. The function $g\lp w \rp$ can also be treated similarly. The residue conditions yield in this case the system 
\begin{align}
M\mathbf{B} = \mathbf{V}.
\end{align} 
Here $M$ is defined as above and 
\begin{align}
V_j := -\kappa_j + H\lp x,t \rp,\quad\quad j=1,2,\ldots,M.
\end{align}
Now, from the ``Residue conditions'' and from (\ref{Aneins}) one has
\begin{align}
H\lp x,t \rp = A_1\lp x,t \rp + A_2\lp x,t \rp +\ldots + A_M\lp x,t \rp
\end{align}
so that both $\mathbf{V}$ and its derivative w.r.t. $x$ are bounded on the whole $(x,t)$-plane. The same is then proved for the vector $\mathbf{B}$, via arguments analogue to the ones above. Finally, the asymptotic estimate (\ref{clas}-\ref{sic}) is a classical result, already available in \cite{WadTod}.
\end{proof}

\subsection{The error vector $\bfe$}

We now wish to estimate the discrepancy between $\bfmt$ and the solution of the (matricial) model Riemann-Hilbert problem $M$. To this purpose, let us introduce the ``error vector''
\begin{align} \label{def_e}
\bfe(w) := \bfmt(w)\cdot\ls M(w) \rs^{-1}. 
\end{align}
We start with a characterization of its following directly from the meromorphic $\dbar$-problem for $\bfmt$. It consists in the following

\begin{SDBP}
Find a 2-dimensional vector-valued function $\bfe$ continuous on the whole complex plane and differentiable with continuity (as a function of two real variables) on $\bbC\backslash \bbR$, such that
\begin{description}
\item[\textbf{i\_}] ``$\dbar$-condition''. For all $z\in\bbC\backslash \bbR$ one has
\begin{align}
\dbar \bfe\lp w \rp = B\lp w \rp \bfe\lp w \rp
\end{align}
where
\begin{align} \label{Belgrado}
B(w)=\small
\lb\begin{array}{ll}
M(w)\,
\scriptsize{\lp \begin{array}{cc} 0 & 0 \\ - [\dbar R(w)]e^{t\Phi(w)} & 0 \end{array} \rp }\,
[M(w)]^{-1} & \mbox{ if } \imag(w)\geq 0,\\
M(w)\,\scriptsize{\lp \begin{array}{cc} 0 & [\dbar R(-w)]e^{-t\Phi(w)} \\ 0 & 0 \end{array} \rp}\, [M(w)]^{-1} & \mbox{ if } \imag(w)\leq 0.
\end{array}\right.\normalsize
\end{align}
\item[\textbf{iii}\_] ``Symmetry condition''. The vector-valued function $\bfe (w)$ is an even function 
\begin{align}\label{simm_e}
\bfe\lp -w \rp = \bfe\lp w \rp,\quad\quad\quad w\in \bbC
\end{align}
\item[\textbf{iv}\_] ``Normalization condition''.
\begin{align}
\lim_{\abs{w}\rightarrow +\infty} \bfe (w) = (1,0).
\end{align}
\end{description}
\end{SDBP}

\begin{remark}
From conditions (\ref{Resuno}-\ref{norm_matricial}) on deduces via simple arguments of complex analysis that 
\begin{align}
\det M\lp w \rp = 2iw.
\end{align}
From formula (\ref{def_e}), it is then not evident that $\bfe$ is smooth in the origin. This follows indeed from property (\ref{simm_e}) together with the observation that $\bfmt$ is differentiable in zero (see remark \ref{Remark_Claudia}).
\end{remark}

Let us now define the operator $\bbJ$ as follows
\begin{align} \label{Mostar}
\ls\bbJ(\bfe)\rs\lp w \rp := -\frac{1}{\pi} \iint_{\bbR^2} 	\frac{\bfe(s)B(s)}{w-s}\de A(s).
\end{align}
The smooth $\dbar$-problem above is easily shown to be equivalent to the following  integral equation
\begin{align} \label{Sarajevo}
\lp \textbf{Id} - \bbJ \rp \bfe = (1,0).
\end{align}
(See \cite{Circle} and \cite{Line} for further details). This is the final reformulation of the Riemann-Hilbert problem \ref{Cesare}, on which we will perform our analysis starting from the next section. 
 
\subsection{Analysis of the integral equation}

In this section we study existence and uniqueness of a solution for equation (\ref{Mostar}-\ref{Sarajevo}). 

\begin{theorem}\label{proposizione_principale}
There exists a constant $C$, depending on $C_0$, such that
\begin{align} \label{stima_inf}
\norm{\bbJ}_{\infty}\leq C\cdot t^{-N+\frac{1}{2}}
\end{align}
for all $t\geq 1$ and all $x\geq C_0 t$.
\end{theorem}

\begin{proof}
Fix such $t$ and $x$ and let $\bfe$ belong to $\lird$. By elementary algebraic manipulations one obtains
\begin{align}
\norm*{\bbJ\lp \bfe \rp}_{\infty} \leq \frac{2\norm{\bfe}_{\infty}}{\pi}\cdot\norm*{\iint_{\bbR^2}\frac{\max_{ij}\abs*{B_{ij}(s)}}{w-s}\de A\lp s \rp}_{\infty}
\end{align}
On $\Omega_3$ the matrix $B$ vanishes, because $\dbar R$ does. On $\Omega_1\cup \Omega_2$ one has
\begin{align}
\abs*{e^{t\Phi\lp u+iv \rp}} \leq e^{-24tvu^2 - C_0tv}.
\end{align}
In view of this last one, of (\ref{diseg_R}) and of (\ref{norm_matricial}) one also has
\begin{align}\label{stima_insieme}
\max_{ij}\abs*{B_{ij}\lp w \rp} \leq C_1\frac{v^N\lp 1+u^2 \rp}{\sqrt{u^2 + v^2}} e^{-24tvu^2 - C_0tv},\quad\quad \quad w\in\Omega_1\cup\Omega_2.
\end{align}
On $\Omega_2$ this last one further simplifies to 
\begin{align}
\max_{ij}\abs*{B_{ij}\lp w \rp} \leq C_2 e^{-24t\delta u^2 - C_0t\delta},\quad\quad\quad w\in \Omega_2.
\end{align}
It follows that
\begin{align}
\iint_{\Omega_2}\frac{\max_{ij}{\abs*{B_{ij}\lp s \rp}}}{\abs*{w-s}} \de A\lp s \rp  & \leq
C_2\cdot e^{-C_0t\delta}\int_\delta^{2\delta}\de b\int_{-\infty}^{+\infty} \frac{e^{-24 t\delta a^2}}{\sqrt{\lp a-u \rp ^2 + \lp b-v \rp^2}} \de a  \\
& \leq C_2\cdot e^{-C_0t\delta}\int_{\delta}^{2\delta} \norm*{e^{- 24 t\delta a^2}}_{L^2\lp \bbR,\de a \rp} \cdot\norm*{\ls \lp u-a \rp^2 + \lp v-b \rp^2 \rs^{-\frac{1}{2}}}_{\ldru} \de b\\
&\leq \frac{C_2}{2} \cdot\sqrt[4]{\frac{\pi^3}{3\delta}}e^{-C_0\delta t}\int_{\delta}^{2\delta} \frac{\de b}{\sqrt{\abs*{b-v}}}\\
& \leq C_3 \cdot\sqrt[4]{\delta}\, e^{-C_0\delta t}.
\end{align}
Let us now consider the region $\Omega_1$. In view of (\ref{stima_insieme}) one has
\begin{align}
\iint_{\Omega_1}\frac{\max_{ij}{\abs*{B_{ij}\lp s \rp}}}{\abs*{w-s}}\de A\lp s \rp &
\leq C_1\cdot  \int_0^{\delta}\de b \int_{-\infty}^{\infty} \frac{b^N\lp 1 + a^2 \rp}{\sqrt{a^2 + b^2} \sqrt{\lp u-a \rp^2 + \lp v-b \rp^2}  }e^{-24tba^2 - C_0tb} \de a\\ 
\nonumber &\leq C_4 \cdot \int_0^{\delta} b^{N-1} e^{-C_0tb} \de b \int_{-\infty}^{+\infty} \frac{e^{-24tba^2}}{\sqrt{\lp u-a \rp^2 + \lp v-b \rp^2}}\de a\\
&\leq C_4 \cdot \int_0^{\delta} b^{N-1} e^{-C_0 t b}\norm*{e^{-24tba^2}}_{L^2\lp \bbR, \de a \rp} \cdot \norm*{\ls \lp u-a \rp^2 + \lp v-b \rp^2 \rs^{-\frac{1}{2}}}_{\ldru}\nonumber\\
& \leq C_5 \cdot \int_0^{\delta} \frac{b^{N-1} e ^{-C_0 tb}}{\sqrt[4]{tb}\sqrt{\abs*{b-v}}} \de b \nonumber \\
%
%
&\leq \frac{C_6}{t^{N-\frac{1}{2}}}.
\end{align}
The regions $\Omega_1^{\prime}, \Omega_2^{\prime}$ and $\Omega_3^{\prime}$ are treated analogously. This completes the proof. 
\end{proof}
A direct consequence of the analysis above is the following, fundamental

\begin{corollary}\label{corollario_sano}
The integral equation (\ref{Mostar}-\ref{Sarajevo}) has a unique solution in $\lird$, whenever $t$ is sufficiently large and $x$ is greater or equal than $C_0 t $. Moreover, for such solution $\bfe = \bfe (w;x,t)$, one has 
\begin{align}
\norm*{\bfe\lp w;x,t \rp}_{\infty} =  (1,0) + \mathcal{O}\lp t^{-N+\frac{1}{2}} \rp,\quad \quad\quad t\rightarrow + \infty;
\end{align} 
uniformly with respect to $x\geq C_0t$.
\end{corollary}
\begin{proof}
In view of (\ref{stima_inf}), one can invert the operator $\mathbb{I}-\bbJ$ by means of Neumann series. This easily yields both parts of the thesis.
\end{proof}

This last result guarantees that the integral equation (\ref{Mostar}-\ref{Sarajevo})  is an equivalent characterization of the ''error vector'' $\bfe$ introduced in (\ref{def_e}). Such an equivalence will be exploited in order to extract as much explicit information as possible about its asymptotic behavior.

\subsection{Long-time asymptotics for the solution $q(x,t)$} \label{Nuova_ossessione}

From the integral equation (\ref{Mostar} - \ref{Sarajevo}) it is easy to deduce that the error vector $\bfe(w)$ has the following asymptotic expansion 
\begin{align}\label{Drina}
\bfe(w;x,t) = (1,0) + \frac{\bfe^{(2)}(x,t)}{w^2} + o\lp\frac{1}{w^2}\rp, \quad\quad w\rightarrow i\infty
\end{align}
where
\begin{align}\label{Sila}
\bfe^{(2)}(x,t) = -\frac{2}{\pi} \iint_{\Omega_1\cup\Omega_2} s\cdot\bfe(s;x,t)\cdot B(s;x,t)\,\de A(s)
\end{align}
Here $w$ is understood to approach infinity along the imaginary axis. In this regime, the nonalaytic vector $\bfmt(w;x,t)$ coincides with $\bfm(w;x,t)$. So (\ref{def_e}) yields
\begin{align} \label{Costarica}
\bfm(w;x,t) = \bfe(w;x,t) \cdot M(w;x,t), \quad\quad w\rightarrow i\infty.
\end{align}
Plugging (\ref{norm_matricial}) and (\ref{Drina}) into (\ref{Costarica}) gives
\begin{align}
m_1^{(1)}(x,t) = iH(x,t) - i e_2^{(2)}(x,t)
\end{align}
for the first component of the vector $\bfm^{(1)}(x,t)$ defined in (\ref{Tisa}). Plugging this last identity into (\ref{Sava}), one obtains
\begin{align} \label{Knopf}
q(x,t) = 2\ddx H(x,t) - 2 \ddx e^{(2)}_2(x,t).
\end{align}
Comparing with (\ref{Nuova_sol}) and recalling (\ref{clas}) yields 
\begin{align}\label{Ciliegio}
\mathcal{E}\lp x,t \rp = -2 \ddx e_2^{(2)}\lp x,t \rp + \mathcal{O}\lp e^{-Ct} \rp
\end{align}
for some constant $C>0$. Estimating the magnitude of this quantity will complete the proof of theorem \ref{Hauptsatz}. To that purpose, we will need the following

\begin{lemma}\label{Ionio}
The solution $\bfe(w;x,t)$ of the integral equation (\ref{Mostar}-\ref{Sarajevo}) is differentiable with respect to $x$ for all positive $x$ and all sufficiently large $t$. There exists a constant $C$, dependent on $C_0$, such that 
\begin{align}\label{Gemona}
\norm*{\ddx \bfe(w;x,t)}_{\infty}\leq C \, t^{-N+\frac{1}{2}} 
\end{align}
for all sufficiently large $t$ and all $x \geq C_0 t$. The $L^{\infty}$-norm above is understood to be computed w.r.t. the complex variable $w$.
\end{lemma}

\begin{proof}[Proof of lemma \ref{Ionio}.]
Solving the integral equation (\ref{Mostar}-\ref{Sarajevo}) by Neumann series, one obtains
\begin{align} \label{Fleons}
\bfe(w;x,t) = \sum_{n=0}^{+\infty}\ls \bbJ^{n}(1,0) \rs (w;x,t).
\end{align}
For each term of the series on the right hand side, the recursive formula 
\begin{align} \label{Fortuna}
\ddx\ls \bbJ^n(1,0) \rs (w;x,t) = \bbJ_x\lb \bbJ^{n-1}\ls (1,0) \rs (w;x,t) \rb + \bbJ\lb \ddx \bbJ^{n-1}\ls (1,0) \rs (w;x,t) \rb, && n\geq 1, 
\end{align}
holds, where
\begin{align}
\ls\bbJ_x \bfe(s)\rs (w) := -\frac{1}{\pi}\iint_{\bbR^2}\frac{\ls\ddx B(s;x,t)\rs \bfe(s)}{w-s} \de A(s) && \bfe \in \lird.
\end{align}
By analogue calculations as in the proof of theorem \ref{proposizione_principale}, (\ref{Fortuna}) yields the estimate
\begin{align}
\norm*{\ddx \ls \bbJ^n (1,0) \rs (w;x,t)}_{\infty} \leq n\cdot \lp \frac{C}{t^{N-\frac{1}{2}}} \rp ^ n, && x\geq C_0 t
\end{align}
valid for some constant $C$, all positive integers and all $t$ sufficiently large. As a consequence,
\begin{align}
\sum_{n=1}^{\infty}\ddx \ls\bbJ^n(1,0)\rs(w;x,t)
\end{align}
converges uniformly in this $(x,t)$-region, and coincides with the derivative with respect to $x$ of the right-hand side of (\ref{Fleons}). This gives differentiability of the left-hand side and estimate (\ref{Gemona}).
\end{proof}

\begin{proof}[Proof of theorem \ref{Hauptsatz}]
The main point here is to control the first term in the right-hand side of (\ref{Ciliegio}). Using (\ref{Sila}), we get for this one the expression 
\begin{align}
-2\ddx e_2^{(2)}(x,t) = \frac{4}{\pi}\ddx \iint_{\Omega_1 \cup \Omega_2 } s\ls  
e_1\sxt B_{12}\sxt + e_2\sxt B_{22}\sxt
\rs \de A(s).
\end{align}
Taking the derivative under the sign of integral and applying lemma \ref{Ionio} and corollary \ref{corollario_sano} one obtains
 \begin{align}\label{Celoria}
 -2\ddx e_2^{(2)}(x,t)  =  \frac{4}{\pi}\iint_{\Omega_1 \cup \Omega_2} s\cdot \ddx B_{12}\sxt \de A(s) + \mathcal{O}(t^{-N+\frac{1}{2}})\cdot I_1
 \end{align}
 where
 \begin{align}
I_1 :=  \iint_{\Omega_1\cup\Omega_2} s \, \lb B_{12}\sxt + B_{22}\sxt + \ddx \ls B_{12}\sxt + B_{22}\sxt \rs\rb \de A(s).
 \end{align}
 Now, in view of lemma \ref{Roseto}, one can determine a positive constant $C_1$ such that
 \begin{align} \label{Banda_semplice}
 \abs*{w B_{1j}(w)}, \abs*{w \ddx B_{1j}(w)} \leq C_1 v^{N}(1+u^2) e^{-24tvu^2-C_0tv}, \quad\quad w\in\Omega_1 \cup \Omega_2;
\end{align}
for $j=1,2$. 
 So that
\begin{align}
I_1 &\leq C_1 \int_0^{2\delta} \de b \int_{-\infty}^{+\infty} b^N\,\lp1 + a^2\rp\, e^{-24tba^2 - C_0tb} \de a \\
& = C_1 \int_0^{2\delta} b^N e^{-C_0 t b} \ls \int_{-\infty}^{+\infty} \lp 1+a^2 \rp e^{-24tba^2} \de a \rs \de b\\
& \leq C_2 \int_{0}^{2\delta} b^N\ls \frac{1}{\sqrt{tb}} + \frac{1}{\lp tb \rp^{\frac{3}{2}}} \rs e^{-C_0 tb} \de b\\
&  = \frac{C_2}{t^{N+1}}\int_0^{+\infty} \bt^N \lp \frac{1}{\sqrt{\bt}} + \frac{1}{\bt^{\frac{3}{2}}} \rp e^{-C_0 \bt} \de \bt\,\, \leq \,\, C_3 t^{-N-1}
\end{align}
Substiting in (\ref{Celoria}) one obtains
 \begin{align}\label{Saldini}
 -2\ddx e_2^{(2)}(x,t)  =  \frac{4}{\pi}\iint_{\Omega_1 \cup \Omega_2} s\cdot \ddx B_{12}\sxt \de A(s) + \mathcal{O}(t^{-2N-\frac{1}{2}})
 \end{align}
 where the asymptotic estimate is to be understood as uniform with respect to $x$ greater or equal than $C_0 t$. Again by means of (\ref{Banda_semplice}) one easily determines a positive constant $C_4$ such that 
 \begin{align}\label{Cile}
 \abs*{\iint_{\Omega_2} s\cdot\ddx B_{12}(s) \de A(s) } \leq e^{-C_4 t},
 \end{align}
 for all $t$ suffciently large and $x$ greater or equal than $C_0 t$. We then turn to analyse, in this same  $(x,t)$-regime,
 \begin{align}\label{Brasile}
 I_2 := \iint_{\Omega_1} s\cdot\ddx B_{12}(s) \de A\lp s \rp.
 \end{align}
An explicit expression for the integrand above is provided by
\begin{align}
s\cdot \ddx B_{12}(s) = \frac{i^{N}}{2N!}r^{(N+1)}(a)\cdot b^N\cdot F(s;x,t)\cdot e^{t\Phi(s)}
\end{align}
where 
\begin{align}
F(s;x,t) = -i f(-s;x,t)\cdot\ddx f(-s;x,t) + s\cdot f(-s;x,t)^2.
\end{align}
and the function $f(w;x,t)$ was defined in (\ref{Aneins}). One can then rewrite (\ref{Brasile}) as follows
\begin{align}\label{Tricomi}
I_2 = \frac{i^N}{2N!}\int_0^{\delta}b^N\de b \int_{-\infty} ^{+\infty}r^{(N+1)}(a)\cdot F(s;x,t)\cdot e^{t\Phi (s)}\de a.
\end{align}
Put 
\begin{align}\label{Kolmogorov}
I_3 := \int_{-\infty}^{\infty}r^{(N+1)}(a)\cdot F(s;x,t)\cdot e^{t\Phi(s)}\de a.
\end{align}
From the original definition (\ref{def_Phi}) of $\Phi$, separating real and imaginary parts,
\begin{align}
\Phi(a+ib) = -24 ba^2 - 2b\lp \xovt - b^2 \rp + 2ia\ls 4a^2 + \lp \xovt - 12b^2 \rp \rs.
\end{align}
In view of our assumptions on the reflection coefficient $r$, there exists a function $\check{r}\in L^1(\bbR)$ such that 
\begin{align}
r^{(N+1)}(u) = \int_{-\infty}^{+\infty}\check{r}(\rho) e^{-i\rho u }\de \rho,\quad\quad\quad u\in \bbR.
\end{align}
Substituting in (\ref{Kolmogorov}) and using Fubini's Theorem, one obtains then 
\begin{align}\label{Cicerone}
I_3 = e^{-2tb\lp \xovt - 4b^2 \rp}\int_{-\infty}^{+\infty}\check{r}(\rho)\lb \int_{-\infty}^{+\infty}  
 \ls F(s;x,t)e^{-24tba^2} \rs e^{ 2ita\ls 4a^2 + \lp \xovt - 12 b^2 \rp - \frac{\rho}{2t} \rs} \de a  \rb \de \rho
\end{align}
Put
\begin{align}
I_4 =  \int_{-\infty}^{+\infty}  
 \ls F(s;x,t)e^{-24tba^2} \rs e^{ 2ita\ls 4a^2 + \lp \xovt - 12 b^2 \rp - \frac{\rho}{2t} \rs} \de a. 
\end{align}
To study this integral we use the van der Corput lemma (see \cite{Ste}, pp 334). We obtain in this way that there exists a (universal) constant $C_{VdC}$ such that 
\begin{align}\label{Bombieri}
\abs*{I_4}\leq C_{VdC}\cdot t^{-\frac{1}{3}}\lb \norm*{F(s;x,t)e^{-24tba^2}}_{L^{\infty}(\bbR,\de a)} + \norm*{ \dda\ls F(s;x,t)e^{-24tba^2} \rs}_{L^1(\bbR,\de a)}   \rb
\end{align} 
 From explicit expression (\ref{Aneins}) one deduces the existence of a constant $C_5$ such that
 \begin{align}
 \abs*{F\lp s;x,t \rp}\leq C_5\lp 1+\abs*{a} \rp, \quad \abs*{\dda F\lp s;x,t \rp}\leq C_5, \quad\quad s\in \Omega_1
 \end{align}
 uniformly for real $x$ and $t$. It immediately follows that 
 \begin{align}
 \norm*{F\lp s;x,t\rp e^{-24tba^2}}_{L^{\infty}\lp \bbR, \de a \rp} \leq C_5 \norm*{\lp 1 + \abs*{a} \rp e^{-24tba^2}}_{L^{\infty}\lp \bbR, \de a \rp}\leq C_{6}\lp 1 + \frac{1}{\sqrt{tb}} \rp
 \end{align}
 and that 
 \begin{align}
\norm*{ \ddu\ls F(s;x,t)e^{-24tba^2} \rs}_{L^1(\bbR,\de a)} & \leq C_{7}\ls \int_{-\infty}^{+\infty} e^{-24tba^2}\de a + \int_{-\infty}^{+\infty} tba^2 e^{-24tba^2} \de a \rs\\
& = C_{7}\ls \frac{1}{\sqrt{tb}}\int_{-\infty}^{\infty}e^{-24\ta^2}\de\ta + \frac{1}{\sqrt{tb}}\int_{-\infty}^{+\infty} \ta^2 e ^{-24\ta^2} \de \ta  \rs \,\leq \,\frac{C_{8}}{\sqrt{tb}}.
\end{align}
Via (\ref{Bombieri}) this yields 
\begin{align}
\abs*{I_4}\leq C_{9}\cdot t^{-\frac{1}{3}}\lp 1+\frac{1}{\sqrt{tb}} \rp.
\end{align}
 By elementary estimates in (\ref{Cicerone}), then,
 \begin{align}
 \abs*{I_3}\leq C_{10}\cdot t^{-\frac{1}{3}}\cdot\lp 1+ \frac{1}{\sqrt{tb}} \rp e^{-2tb\lp \frac{x}{t} - 4b^2 \rp} \leq C_{10}\cdot t^{-\frac{1}{3}}\cdot \lp 1 + \frac{1}{\sqrt{tb}} \rp e^{-C_0tb}.
 \end{align}
Here the constant $C_{10}$ is understood to depend also on the $L^1$-norm of $\check{r}$. This treatment of integral $I_3$ was inspired by \cite{EKT}, lemma 5.1. Finally, substituting according to this last one in (\ref{Tricomi}), one obtains
\begin{align}
\abs*{I_2} &\leq C_{11}\cdot t^{-\frac{1}{3}}\int_0^{\delta} b^{N} \lp 1 + \frac{1}{\sqrt{tb}} \rp e^{-C_0 t b}\de b\\
& \leq C_{11}\cdot  t^{-N-\frac{4}{3}}\int_0^{+\infty} \bt^N\lp 1 + \frac{1}{\sqrt{\bt}} \rp e^{-C_0\bt} \de \bt\, \leq \,C_{12}\cdot t^{-N-\frac{4}{3}}
\end{align}
 Plugging this inequality and (\ref{Cile}) into (\ref{Saldini}) gives the thesis.
 
\end{proof}

\section*{Acknowledgements}

Research partially supported by the Austrian Science Fund (FWF) under Grant No. Y330 and by the INdAM group for mathematical physics. The author wishes to thank Gerald Teschl, Ira Egorova and Spyros Kamvissis for valuable discussions.

\end{document}